\documentclass{llncs}
\usepackage{amssymb,graphicx,epstopdf,multirow}
\newcommand{\com}[1]{}
\newcommand{\bbox}{\rule{0.1in}{0.1in}}

\newenvironment{mycase}[1]{\par\noindent{\itshape Case #1}:\par\noindent\hspace{0.1in}
\begin{tabular}{p{.93\textwidth}}}
{\end{tabular}\par}
\newenvironment{open}{\par\medskip
\noindent\textbf{Open question. }}{\par\medskip}

\begin{document}
\title{The Maximum Clique Problem in Multiple Interval Graphs\thanks{This work was partially supported by the grant ANR-09-JCJC-0041.}}
\author{Mathew C. Francis \and Daniel Gon\c{c}alves \and Pascal Ochem}
\institute{LIRMM, CNRS et Universit\'e Montpellier 2,\\
161 rue Ada 34392 Montpellier Cedex 05, France.\\
\email{\{francis,goncalves,ochem\}@lirmm.fr}}
\date{30 December, 2011}
\maketitle
\pagestyle{plain}
\begin{abstract}
Multiple interval graphs are variants of interval graphs where instead
of a single interval, each vertex is assigned a set of intervals on
the real line.  We study the complexity of the MAXIMUM CLIQUE problem
in several classes of multiple interval graphs.  The MAXIMUM CLIQUE
problem, or the problem of finding the size of the maximum clique, is
known to be NP-complete for $t$-interval graphs when $t\geq 3$ and
polynomial-time solvable when $t=1$. The problem is also known to be
NP-complete in $t$-track graphs when $t\geq 4$ and polynomial-time
solvable when $t\leq 2$. We show that MAXIMUM CLIQUE is already
NP-complete for unit 2-interval graphs and unit 3-track
graphs. Further, we show that the problem is APX-complete for
2-interval graphs, 3-track graphs, unit 3-interval graphs and unit
4-track graphs. We also introduce two new classes of graphs called
$t$-circular interval graphs and $t$-circular track graphs and study
the complexity of the MAXIMUM CLIQUE problem in them.  On the positive
side, we present a polynomial time $t$-approximation algorithm for
WEIGHTED MAXIMUM CLIQUE on $t$-interval graphs, improving earlier work
with approximation ratio $4t$.
\end{abstract}
\pagenumbering{arabic}
\bibliographystyle{plain}
\section{Introduction}
Given a family of sets $\mathcal{F}$, a graph $G$ with vertex set
$V(G)$ and edge set $E(G)$ is said to be an ``intersection graph of
sets from $\mathcal{F}$'' if $\exists f:V(G)\rightarrow \mathcal{F}$
such that for distinct $u,v\in V(G)$, $uv\in E(G)\Leftrightarrow
f(u)\cap f(v)\not=\emptyset$.  When $\mathcal{F}$ is the set of all
closed intervals on the real line, it defines the well-known class of
interval graphs.  A \emph{$t$-interval} is the union of $t$ intervals
on the real line.  When $\mathcal{F}$ is the set of all $t$-intervals,
it defines the class of graphs called \emph{$t$-interval graphs}. This
class was first defined and studied by Trotter and
Harary~\cite{TrotHar}.  Given $t$ parallel lines (or tracks), if each
element of $\mathcal{F}$ is the union of $t$ intervals on different
lines, one defines the class of \emph{$t$-track graphs}.  It is easy
to see that this class forms a subclass of $t$-interval graphs.

These classes of graphs received a lot of attention, for both their
theoretical simplicity and their use in various fields like
Scheduling~\cite{Bar,Hochbaum} or Computational
Biology~\cite{Aumann,Crochemore}.  West and Shmoys~\cite{WestShmoys}
showed that recognizing $t$-interval graphs for $t\geq 2$ is
NP-complete.



Given a circle, the intersection graphs of arcs of this circle forms
the class of \emph{circular arc graphs}.  We introduce similar
generalizations of circular arc graphs. If $G$ has an intersection
representation using $t$ arcs on a circle per vertex, then $G$ is
called a {\em $t$-circular interval} graph. If instead, $G$ has an
intersection representation using $t$ circles and exactly one arc on
each circle corresponding to each vertex of $G$, then $G$ is called a
{\em $t$-circular track} graph. Note that in this case, the class of
$t$-circular track graphs may not be a subclass of the class of
$t$-circular interval graphs.\com{ although we do not know of a
  separating example.}  \com{Clearly, the class of $t$-interval graphs
  forms a subclass of $t$-circular interval graphs and the class of
  unit $t$-interval graphs is a subclass of unit $t$-circular interval
  graphs.}  One can see after cutting the circles, that $t$-circular
interval graphs and $t$-circular track graphs are respectively
contained in $(t+1)$- and $(2t)$-interval graphs.

For all these intersection families of graphs, one can define a
subclass where all the intervals or arcs have the same length. We
respectively call those subclasses \emph{unit $t$-interval},
\emph{unit $t$-track}, \emph{unit $t$-circular interval}, and \emph{unit
  $t$-circular track graphs}.

MAXIMUM WEIGHTED  CLIQUE is the problem of deciding, given a graph $G$
with weighted vertices and an integer $k$, whether $G$ has a clique of
weight $k$.  The case where all the weights are 1 is MAXIMUM CLIQUE.
Zuckerman~\cite{Zuckerman} showed that unless P=NP, there is no
polynomial time algorithm that approximates the maximum clique within
a factor $O(n^{1-\epsilon})$, for any $\epsilon > 0$. MAXIMUM
CLIQUE has been studied for many intersection graphs families. It has
been shown to be polynomial for interval filament
graphs~\cite{Gavril00}, a graph class including circle graphs, chordal
graphs and co-comparability graphs. It has been shown to be
NP-complete for $B_1$-VPG graphs~\cite{MiddPfeiff} (intersection of
strings with one bend and axis-parallel parts~\cite{Asi}), and for
segment graphs~\cite{Cabello} (answering a conjecture of
Kratochv\'{i}l and Ne\v{s}et\v{r}il~\cite{Krat}).
 
MAXIMUM CLIQUE is polynomial for interval graphs (folklore) and for
circular interval graphs~\cite{Gavril73,Hsu}. However, Butman et
al.~\cite{Butman} showed that MAXIMUM CLIQUE is NP-complete for
$t$-interval graphs when $t\geq 3$.  For $t$-track graphs, MAXIMUM
CLIQUE is polynomial-time solvable when $t\leq 2$ and NP-complete when
$t\geq 4$~\cite{Koenig}.  Butman et al. also showed a polynomial-time
$\frac{t^2-t+1}{2}$ factor approximation algorithm for MAXIMUM CLIQUE
in $t$-interval graphs. Koenig~\cite{Koenig} observed that a similar
approximation algorithm with a slightly better approximation ratio
$\frac{t^2-t}{2}$ exists for MAXIMUM CLIQUE in $t$-track
graphs. Butman et al. asked the following questions:
\begin{itemize}
\item Is MAXIMUM CLIQUE NP-hard in 2-interval graphs?
\item Is it APX-hard in $t$-interval graphs for any constant $t\geq 2$?
\item Can an algorithm with a better approximation ratio than
  $\frac{t^2-t+1}{2}$ be achieved for $t$-interval graphs?
\end{itemize}
We answer all of these questions in the affirmative. As far as the
third question is concerned, Kammer, Tholey and Voepel~\cite{Kammer}
have already presented an improved polynomial-time approximation
algorithm that achieves an approximation ratio of $4t$ for
$t$-interval graphs. In this paper (Section~\ref{sec:approx}), we
present a linear time $2t$-approximation algorithm, and a polynomial
time $t$-approximation algorithm for MAXIMUM WEIGHTED CLIQUE in
$t$-interval graphs (and thus in $t$-track graphs), $t$-circular
interval graphs, and $t$-circular track graphs.  Then we show in
Section~\ref{sec:apx-h} that MAXIMUM CLIQUE is APX-complete for many
of these families (including 2-interval graphs). In
Section~\ref{sec:np-h}, we show that for some of the remaining classes
(including unit 2-interval graphs) MAXIMUM CLIQUE is NP-complete.  In
Section~\ref{sec:comp} we give some APX-hardness results for several
problems restricted to the complement class of $t$-interval graphs.
Finally, we conclude with some remarks and open questions.

\section{Preliminaries}

Consider a circle $C$ of length $l$ with a distinguished point $O$.
The coordinate of a point $p\in C$ is the length of the arc going
clockwise from $O$ to $p$. Given two reals $p$ and $q$, $[p,q]$ is the
arc of $C$ going clockwise from the point with coordinate $p$ to the
one with coordinate $q$. In the following, coordinates
are understood modulo $l$.

A \emph{representation} of a $t$-interval graph $G$ is a set of $t$
functions, $I_1,\ldots,I_t$, assigning each vertex in $V(G)$ to an
interval of the real line. For $t$-track graphs we have $t$ lines
$L_1,\ldots,L_t$, and each $I_i$ assigns intervals from $L_i$.
Similarly, for a representation of $t$-circular interval graphs
(resp. $t$-circular track graphs) we have a circle $C$ (resp. $t$
circles $C_1,\ldots,C_t$) and $t$ functions $I_i$, assigning each
vertex in $V(G)$ to an arc of $C$ (resp. of $C_i$).

\section{Approximation algorithms}
\label{sec:approx}

The first approximation algorithms for the MAXIMUM CLIQUE in
$t$-interval graphs and $t$-track graphs~\cite{Butman,Koenig} are
based on the fact that any $t$-interval representation
(resp. $t$-track representation) of a clique admits a transversal
(i.e. a set of points touching at least one interval of each vertex)
of size $\tau = t^2-t+1$ (resp. $\tau =
t^2-t$)~\cite{Kaiser}. Scanning the representation of a graph $G$ from
left to right (in time $O(tn)$) one passes through the points of the
transversal of a maximum clique $K$ of $G$. At some of those points
there are at least $|K|/\tau$ intervals forming a subclique of
$K$. Thus, this gives an $O(tn)$-time $\tau$-approximation. Butman et
al. improved this ratio by 2 by considering every pair of points in
the representation. The intervals at these points induce a
co-bipartite graph, for which computing the maximum clique is
polynomial (as computing a maximum independent set of a bipartite
graph is polynomial).  Then one can see that this gives a polynomial
time $(\tau/2)$-approximation algorithm.  This actually gives a
polynomial exact algorithm for the MAXIMUM CLIQUE in $2$-track
graphs~\cite{Koenig}, as $\tau =2$ in this case.  For the other cases,
Kammer et al.~\cite{Kammer} greatly improved the approximation ratios
from roughly $t^2/2$ to $4t$, using the new notion of $k$-perfect
orientability.  Using transversal arguments, we can easily improve
this ratio for some subclasses.  A representation is \emph{balanced}
if for each vertex, all its intervals (or arcs) have the same length.
\begin{remark}
In any balanced $t$-interval (resp. $t$-track, $t$-circular interval,
or $t$-circular track) representation of a clique, the $2t$ interval
extremities of the vertex with the smallest intervals form a
transversal.  Thus, in those classes of graphs MAXIMUM CLIQUE admits a
linear time $2t$-approximation algorithm, and a polynomial time
$t$-approximation algorithm.
\end{remark}

We shall now show how to achieve the same approximation ratio without
restraining to balanced representations.

\begin{theorem}\label{thm:approx}
There is a linear time $2t$-approximation algorithm, and a polynomial
time $t$-approximation algorithm for MAXIMUM WEIGHTED  CLIQUE on
$t$-interval graphs, $t$-track graphs, $t$-circular interval graphs,
and $t$-circular track graphs.
\end{theorem}
\begin{proof}
The problem is polynomial when $t=1$, we thus assume that $t\ge 2$.
Let us prove the theorem for $t$-interval graphs, the proofs for the
other classes are exactly the same. Let $G$ be a weighted $t$-interval
graph with weight function $w(u)$ on its vertices, and let $K$ be a
maximum weighted clique of $G$.  Let $I_1,\ldots,I_t$ form a
$t$-interval representation of $G$ such that for any vertex $u\in
V(G)$, $I_i(u)=[u_i,u'_i]$. For any edge $uv$ there exists a $i$ and a
$j\in[t]$ such that the point $u_i$ belongs to $I_j(v)$, or such that
$v_j\in I_i(u)$. One can thus orient and color the edges of $G$ in
such a way that $uv$ goes from $u$ to $v$ in color $i$ if $u_i\in
I_j(v)$ for some $j$. In $K$ there is a vertex $u$ with more weight on
its out-neighbors in $K$ than on its in-neighbors in $K$. Indeed, this
comes from the fact that in the oriented graphs obtained from $K$ by
replacing each vertex $u$ by $w(u)$ vertices $u_i$ and by putting an
arc $u_iv_j$ if and only if there is an arc $uv$ in $K$, there is a
vertex $u_i$ with $d^+(u_i)\ge d^-(u_i)$, which is equivalent to
$w(N^+_K(u))\ge w(N^-_K(u))$.  Thus there exists two distinct values
$i$ and $j$ such that $u$ has at least weight $(w(K)-w(u))/2t$ on its
out-neighbors in color $i$, and at least $(w(K)-w(u))/t$ out-neighbors
in color $i$ or $j$. The vertex $u$ and its out-neighbors in a given
color clearly induce a clique of $G$ (they intersect at $u_i$). Thus
scanning the representation from left to right looking for the point
with the more weights gives a clique of weight at least
$w(u)+(w(K)-w(u))/2t > w(K)/2t$, which is a $2t$-approximation.

Then the graph induced by $u$ and its out-neighbors in color $i$ or
$j$ being co-bipartite one can compute its maximum weighted clique in
polynomial time (as computing a maximum weighted independent set of a
bipartite graph is polynomial). This clique has weight at least
$w(u)+(w(K)-w(u))/t > w(K)/t$ (the weight of the subclique of $K$
induced by $u$ and its neighbors in color $i$ or $j$).  Thus, for each
vertex $u$ of the graph and any pair $u_i$ and $u_j$ of interval left
end, if we compute the maximum weighted clique of the corresponding
co-bipartite graph, we obtain a $t$-approximation.
\end{proof}

\section{APX-hardness in multiple interval graphs}
\label{sec:apx-h}

The complement of a graph $G$ is denoted by $\overline{G}$.  Given a
graph $G$ on $n$ vertices with $V(G)=\{x_1,\ldots,x_n\}$ and
$E(G)=\{e_1,\ldots, e_m\}$, and a positive integer $w$, we define
$Subd_w(G)$ to be the graph obtained by subdividing each edge of $G$
$w$ times. If $e_k\in E(G)$ and $e_k=x_ix_j$ where $i<j$, we define
$l(k)=i$ and $r(k)=j$ (as if $x_i$ and $x_j$ were respectively the
left and the right end of $e_k$). In the following we subdivide edges
2 or 4 times.  In $Subd_2(G)$ (resp. $Subd_4(G)$), the vertices
subdividing $e_k$ are $a_k$ and $b_k$ (resp. $a_k, b_k, c_k,$ and
$d_k$) and they are such that $(x_{l(k)},a_k,b_k,x_{r(k)})$
(resp. $(x_{l(k)},a_k,b_k,c_k,d_k,x_{r(k)})$) is the subpath of
$Subd_2(G)$ (resp. $Subd_4(G)$) corresponding to $e_k$.
To prove APX-hardness results we need the following structural theorem,
which is of independent interest.
\begin{theorem}\label{thm:subd}
Given any graph $G$, 
\begin{itemize}
\item $\overline{Subd_4(G)}$ is a 2-interval graph,
\item $\overline{Subd_2(G)}$ is a unit 3-interval graph,
\item $\overline{Subd_2(G)}$ is a 3-track graph,
\item $\overline{Subd_2(G)}$ is a unit 4-track graph,
\item $\overline{Subd_2(G)}$ is a unit 2-circular interval graph (and thus a 2-circular interval graph),
\item $\overline{Subd_2(G)}$ is a 2-circular track graph, and
\item $\overline{Subd_2(G)}$ is a unit 4-circular track graph.
\end{itemize}
Furthermore, such representations can be constructed in linear time.
\end{theorem}

Since MAXIMUM INDEPENDENT SET is APX-hard even when restricted to
degree bounded graphs~\cite{PapaYanna,BermanFujito}, Chleb\'ik and
Chelb\'ikov\'a~\cite{ChlebChleb} observed that MAXIMUM INDEPENDENT SET
is APX-hard even when restricted to $2k$-subdivisions of 3-regular
graphs for any fixed integer $k\ge 0$. Taking the complement graphs, we
thus have that MAXIMUM CLIQUE is APX-hard even when restricted to the
set ${\mathcal C}_{2k} = \{\overline{Subd_{2k}(G)}\ |$ any graph
$G\}$, for any fixed integer $k\ge 0$. Thus, since MAXIMUM CLIQUE is
approximable for all the graph classes considered in
Theorem~\ref{thm:subd}, we clearly have the next result.

\begin{theorem}\label{thm:apx}
MAXIMUM CLIQUE is APX-complete for:
\begin{itemize}
\item 2-interval graph,
\item unit 3-interval graph,
\item 3-track graph,
\item unit 4-track graph,
\item unit 2-circular interval graph (and thus for 2-circular interval graphs), 
\item 2-circular track graph, and
\item unit 4-circular track graph.
\end{itemize}
\end{theorem}

\begin{remark}
To prove that MAXIMUM CLIQUE is NP-hard on $B_1$-VPG graphs,
Middendorf and Pfeiffer~\cite{MiddPfeiff} proved that for any graph
$G$, $\overline{Subd_2(G)} \in B_1$-VPG. One can thus see that MAXIMUM
CLIQUE is actually APX-hard for this class of graphs.
\end{remark}

We prove Theorem~\ref{thm:subd} in the following subsections.

\subsection{2-interval graphs}
\begin{theorem}\label{thm:2-int}
Given any graph $G$, $\overline{Subd_4(G)}$ is a 2-interval graph and a 2-interval
representation for it can be constructed in linear time.
\end{theorem}
\begin{proof}
Recall that each edge $e_k=x_ix_j$ of $G$ where $i<j$, corresponds to
the path $(x_i,a_k,b_k,c_k,d_k,x_j)$ in $Subd_4(G)$.  We define the
representation $\{I_1,I_2\}$ of $\overline{Subd_4(G)}$ as follows (see
also Figure~\ref{fig:2i}). For $1\leq i\leq n$ and $1\leq k\leq m$:

\medskip

$
\begin{array}{lcl}
I_1(a_k)&=&[0,m(l(k)-1)+k-1]\\
I_1(x_i)&=&[mi,mn+mi]\\
I_2(a_k)&=&[mn+ml(k)+1,4mn+m-ml(k)-k+1]\\
I_1(b_k)&=&[m(l(k)-1)+k,mn+m-k]\\
I_1(c_k)&=&[mn+m-k+1,3mn+m-mr(k)-k+1]\\
I_1(d_k)&=&[3mn+m-mr(k)-k+2,4mn+mr(k)]\\
I_2(b_k)&=&[4mn+m-ml(k)-k+2,5mn+k]\\
I_2(x_i)&=&[4mn+mi+1,5mn+mi+1]\\
I_2(d_k)&=&[5mn+mr(k)+k+1,6mn+m+1]\\
I_2(c_k)&=&[5mn+k+1,5mn+mr(k)+k]
\end{array}
$

\medskip
\begin{figure}
\center
\includegraphics[scale=0.6]{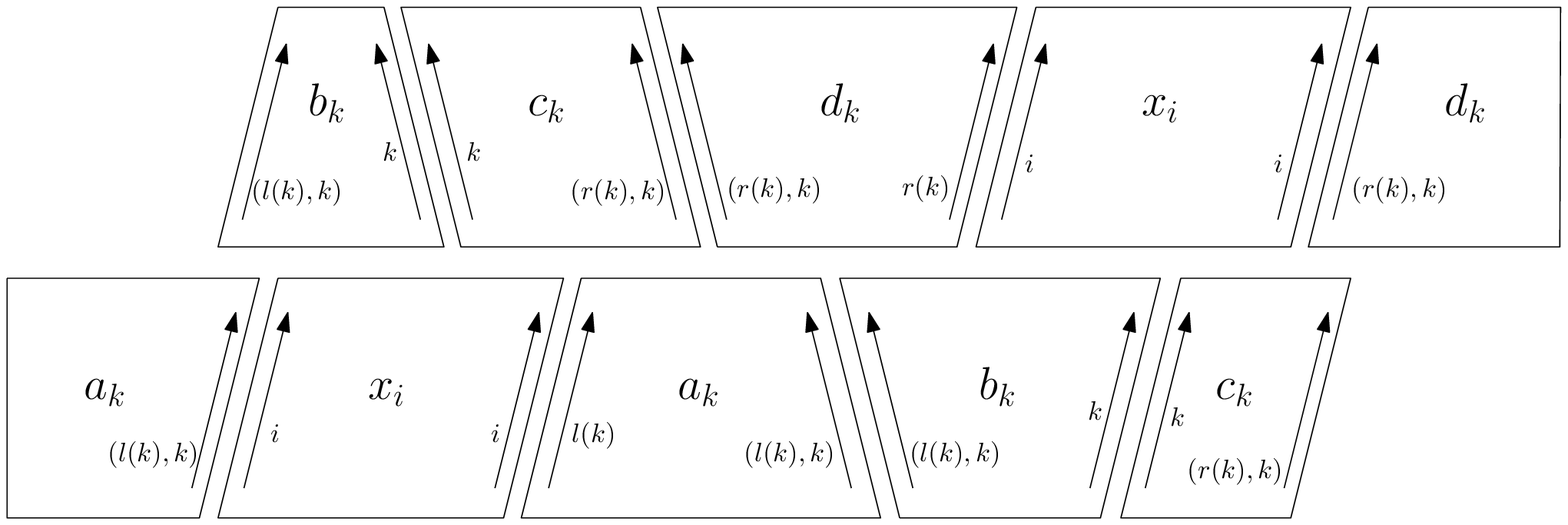}
\caption{The 2-interval representation of $\overline{Subd_4(G)}$.}
\label{fig:2i}
\end{figure}

Figure~\ref{fig:2i} (and the other figures of this kind) should be
understood in the following way.  The leftmost block labeled $a_k$
corresponds to the intervals $I_1(a_k)$, and its shape, together with
the label $(l(k),k)$ on the arrow mean that,
\begin{itemize}
\item the left end of the intervals $I_1(a_k)$ are the same (coordinate 0), and that
\item the right end of the intervals $I_1(a_k)$ are ordered (from left to right)
  accordingly to $l(k)$, and in case of equality, accordingly to $k$.
\end{itemize}
Here we can see that this block is close to the blocks $I_1(b_k)$, and
$I_1(x_i)$.  

The left end of the interval $I_1(b_k)$ is also ordered
(from left to right) accordingly to $(l(k),k)$. Such situation means
that $I_1(a_k)$ intersects every $I_1(b_{k'})$ such that $(l(k),k) >
(l(k'),k')$, i.e. such that $l(k) > l(k')$ or such that $l(k)=l(k')$
and $k>k'$. Note that since, between $I_2(a_k)$ and $I_2(b_k)$ we have
the opposite situation, for any vertex $a_k$, $a_k$ is adjacent to
every $b_{k'}$, except $b_k$.

The left end of the interval $I_1(x_i)$ is ordered
(from left to right) accordingly to $i$. Such situation means
that $I_1(a_k)$ intersects every $I_1(x_i)$ such that $l(k) > i$.
Note that since, between $I_1(x_i)$ and $I_2(a_k)$ we have
the opposite situation, for any vertex $a_k$, $a_k$ is adjacent to
every $x_i$, except $x_{l(k)}$.

We claim that $I_1$ and $I_2$ together form a valid 2-interval
representation for $\overline{Subd_4(G)}$. One can check it with
Figure~\ref{fig:2i}, but we give a full proof for this first
construction. For any two vertices $u$ and $v$ of
$\overline{Subd_4(G)}$, we will show that $uv$ is an edge of
$\overline{Subd_4(G)}$ if and only if $I_1(u)\cup I_2(u)$ intersects
$I_1(v)\cup I_2(v)$. We first consider the case where $uv$ is an edge.

\medskip

\begin{mycase}{$u=x_i$ and $v=x_j$} $[mn,mn+m]\subseteq I_1(x_i) \cap
I_1(x_j)$.\end{mycase}
\begin{mycase}{$u=x_i$ and $v=a_k$, where $l(k)\not=i$} If $l(k)>i$,
then $mi\in I_1(a_k)\cap I_1(x_i)$. If on
the other hand, $l(k)<i$, then $mn+mi\in I_1(x_i)\cap I_2(a_k)$.\end{mycase}
\begin{mycase}{$u=x_i$ and $v=b_k$} $mn\in I_1(x_i)\cap I_1(b_k)$.\end{mycase}
\begin{mycase}{$u=x_i$ and $v=c_k$} $mn+m\in I_1(x_i)\cap I_1(c_k)$.\end{mycase}
\begin{mycase}{$u=x_i$ and $v=d_k$, where $r(k)\not=i$} If $r(k)>i$,
then $4mn+mi+m\in I_1(d_k)\cap I_2(x_i)$
and if $r(k)<i$, then $5mn+mi+1\in I_2(x_i)\cap I_2(d_k)$.\end{mycase}
\begin{mycase}{$u=a_k$ and $v=a_{k'}$} $0\in I_1(a_k)\cap
I_1(a_{k'})$.\end{mycase}
\begin{mycase}{$u=a_k$ and $v=b_{k'}$, where $k\not=k'$} If
$l(k')<l(k)$, then $m(l(k)-1)\in I_1(a_k)\cap I_1(b_{k'})$
and if $l(k)<l(k')$, then $4mn-ml(k)+1\in I_2(a_k)\cap I_2(b_{k'})$.
Suppose $l(k)=l(k')$. Now, if $k'<k$,
then $m(l(k)-1)+k-1\in I_1(a_k)\cap I_1(b_{k'})$ and if $k'>k$, then
$4mn+m-ml(k)-k+1\in I_2(a_k)\cap I_2(b_{k'})$.
\end{mycase}
\begin{mycase}{$u=a_k$ and $v=c_{k'}$} $2mn+1\in I_2(a_k)\cap
I_1(c_{k'})$.\end{mycase}
\begin{mycase}{$u=a_k$ and $v=d_{k'}$} $3mn+1\in I_2(a_k)\cap
I_1(d_{k'})$.\end{mycase}
\begin{mycase}{$u=b_k$ and $v=b_{k'}$} $mn\in I_1(b_k)\cap
I_1(b_{k'})$.\end{mycase}
\begin{mycase}{$u=b_k$ and $v=c_{k'}$, where $k\not=k'$} If $k<k'$,
then $mn+m-k\in I_1(b_k)\cap I_1(c_{k'})$.
\end{mycase}
\begin{mycase}{$u=b_k$ and $v=d_{k'}$} $4mn+1\in I_2(b_k)\cap
I_1(d_{k'})$.\end{mycase}
\begin{mycase}{$u=c_k$ and $v=c_{k'}$} $[mn+m,2mn+1]\subseteq
I_1(c_k)\cap I_1(c_{k'})$.\end{mycase}
\begin{mycase}{$u=c_k$ and $v=d_{k'}$, where $k\not=k'$} If
$r(k)<r(k')$, then $3mn-mr(k)+1\in I_1(c_k)\cap I_1(d_{k'})$
and if $r(k')<r(k)$, then $5mn+mr(k)+1\in I_2(c_k)\cap I_2(d_{k'})$.
Suppose $r(k)=r(k')$. Now, if $k<k'$, $3mn-mr(k)+1
\in I_1(c_k)\cap I_1(d_{k'})$ and if $k'<k$, then $5mn+mr(k)+k\in
I_2(c_k)\cap I_2(d_{k'})$.\end{mycase}
\begin{mycase}{$u=d_k$ and $v=d_{k'}$} $6mn+m+1\in I_2(d_k)\cap
I_2(d_{k'})$.\end{mycase}

\medskip

Let us now consider the case where $uv$ is not an edge.  In
particular, let us show that $I_1(u)<I_1(v)<I_2(u)<I_2(v)$, where
$[u,u']<[v,v']$ means that $u'<v$.

\medskip

\begin{mycase}{$u=x_i$ and $v=a_k$, where $l(k)=i$}
$I_1(a_k)<I_1(x_i)<I_2(a_k)<I_2(x_i)$.\end{mycase}
\begin{mycase}{$u=x_i$ and $v=d_k$, where $r(k)=i$}
$I_1(x_i)<I_1(d_k)<I_2(x_i)<I_2(d_k)$.\end{mycase}
\begin{mycase}{$u=a_k$ and $v=b_k$}
$I_1(a_k)<I_1(b_k)<I_2(a_k)<I_2(b_k)$.\end{mycase}
\begin{mycase}{$u=b_k$ and $v=c_k$}
$I_1(b_k)<I_1(c_k)<I_2(b_k)<I_2(c_k)$.\end{mycase}
\begin{mycase}{$u=c_k$ and $v=d_k$}
$I_1(c_k)<I_1(d_k)<I_2(c_k)<I_2(d_k)$.\end{mycase}

\medskip

Therefore, we have a valid 2-interval representation of
$\overline{Subd_4(G)}$ and this representation can obviously be
constructed in linear time.\hfill\bbox
\end{proof}

\subsection{Unit 3-interval graphs}
\begin{theorem}\label{thm:unit-3-int}
Given any graph $G$, $\overline{Subd_2(G)}$ is a unit 3-interval graph and a unit 3-interval
representation for it can be constructed in linear time.
\end{theorem}
\begin{proof}
Recall that each edge $e_k=x_ix_j$ of $G$ where $i<j$, corresponds to
the path $(x_i,a_k,b_k,x_j)$ in $Subd_2(G)$.  We define $I_1$, $I_2$
and $I_3$ as follows (see also Figure~\ref{fig:u3i}). Here again,
$1\leq i\leq n$ and $1\leq k\leq m$.

\medskip

$\begin{array}{lcl}
I_1(b_k)&=&[m(l(k)-1)+k,m(l(k)-1)+m^2+k]\\
I_1(a_k)&=&[m(l(k)-1)+m^2+k+1,m(l(k)-1)+2m^2+k+1]\\
I_1(x_i)&=&[mi+2m^2+2,mi+3m^2+2]\\
I_2(b_k)&=&[mr(k)+3m^2+k+2,mr(k)+4m^2+k+2]\\
I_2(x_i)&=&[mi+4m^2+m+3,mi+5m^2+m+3]\\
I_2(a_k)&=&[ml(k)+5m^2+m+k+3,ml(k)+6m^2+m+k+3]\\
I_3(b_k)&=&[ml(k)+6m^2+m+k+4,ml(k)+7m^2+m+k+4]\\
I_3(a_k)&=&[15m^2,16m^2]\\
I_3(x_i)&=&[17m^2,18m^2]
\end{array}
$

\medskip
\begin{figure}
\center
\includegraphics[scale=0.6]{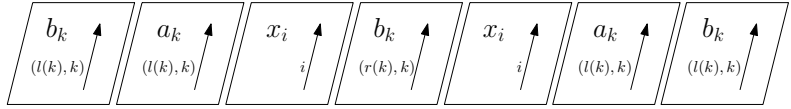}
\caption{The unit 3-interval representation of $\overline{Subd_2(G)}$.}
\label{fig:u3i}
\end{figure}

This representation can be constructed in linear time and it is easy
to verify that $I_1$, $I_2$ and $I_3$ assign intervals of length $m^2$
to the vertices of $\overline{Subd_2(G)}$. Then one can also easily
check in the figure that this is a valid unit 3-interval
representation of $\overline{Subd_2(G)}$.\hfill\bbox
\end{proof}

\subsection{3-track graphs}
\begin{theorem}\label{thm:3-track}
Given any graph $G$, $\overline{Subd_2(G)}$ is a 3-track graph and a
3-track representation for it can be constructed in linear time.
\end{theorem}
\begin{proof}
We define a 3-track representation for $\overline{Subd_2(G)}$ as
follows (see also Figure~\ref{fig:3t}). For $1\leq i\leq n$ and $1\leq
k\leq m$:

\medskip

$
\begin{array}{lcl}
I_1(a_k)&=&[0,l(k)]\\
I_1(x_i)&=&[i+1,n+i+1]\\
I_1(b_k)&=&[n+r(k)+2,2n+3]\\
I_2(x_i)&=&[0,i]\\
I_2(a_k)&=&[l(k)+1,n+k]\\
I_2(b_k)&=&[n+k+1,m+n+2]\\
I_3(a_k)&=&[0,m+1-k]\\
I_3(b_k)&=&[m+2-k,m+r(k)]\\
I_3(x_i)&=&[m+i+1,m+n+2]
\end{array}
$

\medskip
\begin{figure}
\center
\includegraphics[scale=0.6]{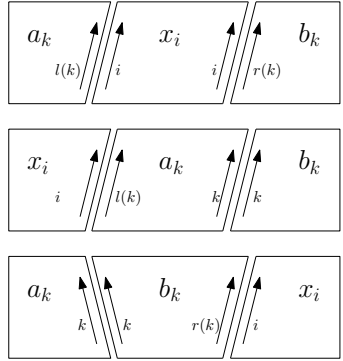}
\caption{The 3-track representation of $\overline{Subd_2(G)}$.}
\label{fig:3t}
\end{figure}

This representation can be constructed in linear time and one can
easily check in the figure that this is a valid 3-track representation
of $\overline{Subd_2(G)}$.\hfill\bbox
\end{proof}

\subsection{Unit 4-track graphs}
\begin{theorem}\label{thm:unit-4-track}
Given any graph $G$, $\overline{Subd_2(G)}$ is a unit 4-track graph and a unit 4-track
representation for it can be constructed in linear time.
\end{theorem}
\begin{proof}
We define $I_1$, $I_2$, $I_3$ and $I_4$ as follows (see
also Figure~\ref{fig:u4t}). As usual, $1\leq i\leq n$
and $1\leq k\leq m$.

\medskip

$\begin{array}{lcl}
I_1(a_k)&=&[m(l(k)-1)+k,m(l(k)-1)+m^2+k]\\
I_1(x_i)&=&[mi+m^2+1,mi+2m^2+1]\\
I_1(b_k)&=&[2m^2+mr(k)+k+1,3m^2+mr(k)+k+1]\\
I_2(b_k)&=&[m(r(k)-1)+k,m(r(k)-1)+m^2+k]\\
I_2(x_i)&=&[mi+m^2+1,mi+2m^2+1]\\
I_2(a_k)&=&[2m^2+ml(k)+k+1,3m^2+ml(k)+k+1]\\
I_3(a_k)&=&[k,k+m^2]\\
I_3(b_k)&=&[k+m^2+1,k+2m^2+1]\\
I_3(x_i)&=&[5m^2,6m^2]\\
I_4(b_k)&=&[k,k+m^2]\\
I_4(a_k)&=&[k+m^2+1,k+2m^2+1]\\
I_4(x_i)&=&[5m^2,6m^2]
\end{array}
$

\medskip
\begin{figure}
\center
\includegraphics[scale=0.6]{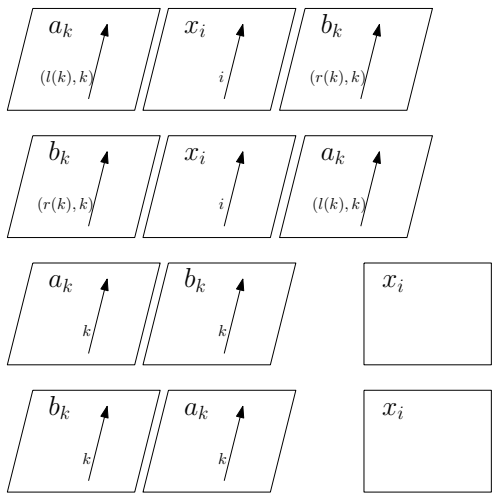}
\caption{The unit 4-track representation of $\overline{Subd_2(G)}$.}
\label{fig:u4t}
\end{figure}

This representation can be constructed in linear time and it is easy
to verify that $I_1$, $I_2$, $I_3$ and $I_4$ assign intervals of length $m^2$
to the vertices of $\overline{Subd_2(G)}$. Then one can also easily
check in the figure that this is a valid unit 4-track
representation of $\overline{Subd_2(G)}$.\hfill\bbox
\end{proof}

\subsection{Unit 2-circular interval graphs}
\begin{theorem}\label{thm:unit-2-circint}
Given any graph $G$, $\overline{Subd_2(G)}$ is a unit 2-circular interval graph
and a unit 2-circular interval representation for it can be constructed in
linear time.
\end{theorem}
\begin{proof}
Let $C$ be a circle of circumference $6m^2+2m+4$. The mappings $I_1$
and $I_2$, which map $V(G)$ to arcs on this circle, are defined as
follows (see also Figure~\ref{fig:u2ci}).

\medskip

$
\begin{array}{lcl}
I_1(b_k)&=&[ml(k)+6m^2+m+k+4,m(l(k)-1)+m^2+k]\\
I_1(a_k)&=&[m(l(k)-1)+m^2+k+1,m(l(k)-1)+2m^2+k+1]\\
I_1(x_i)&=&[mi+2m^2+2,mi+3m^2+2]\\
I_2(b_k)&=&[mr(k)+3m^2+k+2,mr(k)+4m^2+k+2]\\
I_2(x_i)&=&[mi+4m^2+m+3,mi+5m^2+m+3]\\
I_2(a_k)&=&[ml(k)+5m^2+m+k+3,ml(k)+6m^2+m+k+3]
\end{array}
$

\medskip
\begin{figure}
\center
\includegraphics[scale=0.6]{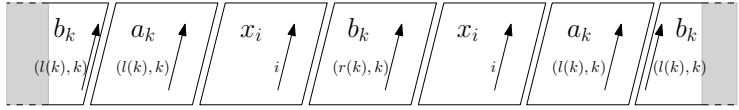}
\caption{The unit 2-circular interval representation of $\overline{Subd_2(G)}$.}
\label{fig:u2ci}
\end{figure}

Note that this representation is almost the same as the unit 3-interval
representation given for $\overline{Subd_2(G)}$ in the proof of
Theorem~\ref{thm:unit-3-int}, the only difference being that $I_1(b_k)$ and
$I_3(b_k)$ have now been fused to form $I_1(b_k)$ of the unit 2-circular
interval representation being constructed.
This representation can be constructed in linear time and it is easy
to verify that the arcs have length $m^2$. Then one can also easily
check in the figure that this is a valid unit 2-circular interval
representation of $\overline{Subd_2(G)}$.\hfill\bbox
\end{proof}

\com{

It can be easily verified that $I_1$ and $I_2$ map the vertices of
$\overline{Subd_2(G)}$ to arcs of length $m^2$ on the circle $C$ and that two
vertices $u$ and $v$ of $\overline{Subd_2(G)}$ are adjacent if and only if $\exists
i,j$ such that $I_i(u)\cap I_j(v)\not=\emptyset$. Thus, $I_1$, $I_2$ and $C$
form a valid unit 2-circular interval representation of $\overline{Subd_2(G)}$ that
can be constructed in linear time.
\hfill\bbox
\end{proof}
}

\subsection{2-circular track graphs}
\begin{theorem}\label{thm:2-circtrack}
Given any graph $G$, $\overline{Subd_2(G)}$ is a 2-circular track graph and a
2-circular track representation for it can be constructed in linear time.
\end{theorem}
\begin{proof}
We define a 2-circular track representation using circles $C_1$ and
$C_2$, each having circumference at least $3n+1$, and mappings $I_1$
and $I_2$ defined as follows (see also Figure~\ref{fig:2ct}).

\medskip

$
\begin{array}{lcl}
I_1(x_i)=[i,i+n]\\
I_1(a_k)=[l(k)+n+1,l(k)+2n]\\
I_1(b_k)=[l(k)+2n+1,r(k)-1]\\
I_2(x_i)=[i,i+n]\\
I_2(b_k)=[r(k)+n+1,r(k)+2n]\\
I_2(a_k)=[r(k)+2n+1,l(k)-1]
\end{array}
$

\medskip
\begin{figure}
\center
\includegraphics[scale=0.6]{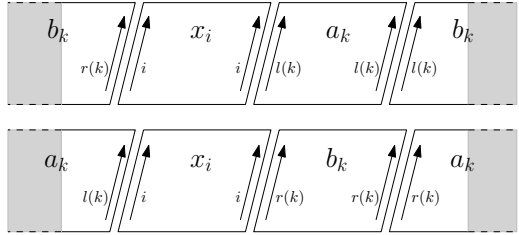}
\caption{The 2-circular track representation of $\overline{Subd_2(G)}$.}
\label{fig:2ct}
\end{figure}

Clearly, this representation can be constructed in linear time, and as
before, it can be checked that the circles $C_1$ and $C_2$ together
with the mappings $I_1$ and $I_2$ form a valid 2-circular track
representation of $\overline{Subd_2(G)}$.\hfill\bbox
\end{proof}

\section{NP-hardness in unit 2-interval and unit 3-track graphs}
\label{sec:np-h}

\com{
Let us first define some terminology for this section.
We refer to $Subd_w(G)$ as a \emph{$w$-subdivision} of $G$. Given a graph $G$,
a graph obtained by subdividing each of its edges an odd number of times (or
equivalently, replacing each edge with a path of even length) is called an
\emph{odd subdivision} of $G$. Similarly, the graph obtained by subdividing
each edge of $G$ an even number of times (or equivalently, replacing each
edge with an odd length path) is called an \emph{even subdivision} of $G$.
}

Valiant~\cite{Valiant} has shown that every planar graph of degree at
most 4 can be drawn on a grid of linear size such that the vertices
are mapped to points of the grid and the edges to piecewise linear
curves made up of horizontal and vertical line segments whose
endpoints are also points of the grid. It is immediately clear that
every planar graph $G$ has a subdivision $G'$ that is an induced
subgraph of a grid graph such that each edge of $G$ corresponds to a
path of length at most $O(|V(G)|^2)$ (see Figure~\ref{fig:embed}).
Note that here, some paths have even length and some have odd length.
An {\em even subdivision} (resp. {\em odd subdivision}) of $G$ is a
graph obtained from $G$ by subdividing each edge $e$ of $G$ an even
(resp. odd) number of times, and at most $|V(G)|^{O(1)}$ times.

\begin{figure}
\center
\begin{tabular}{cp{0.5in}c}
\includegraphics[scale=0.5]{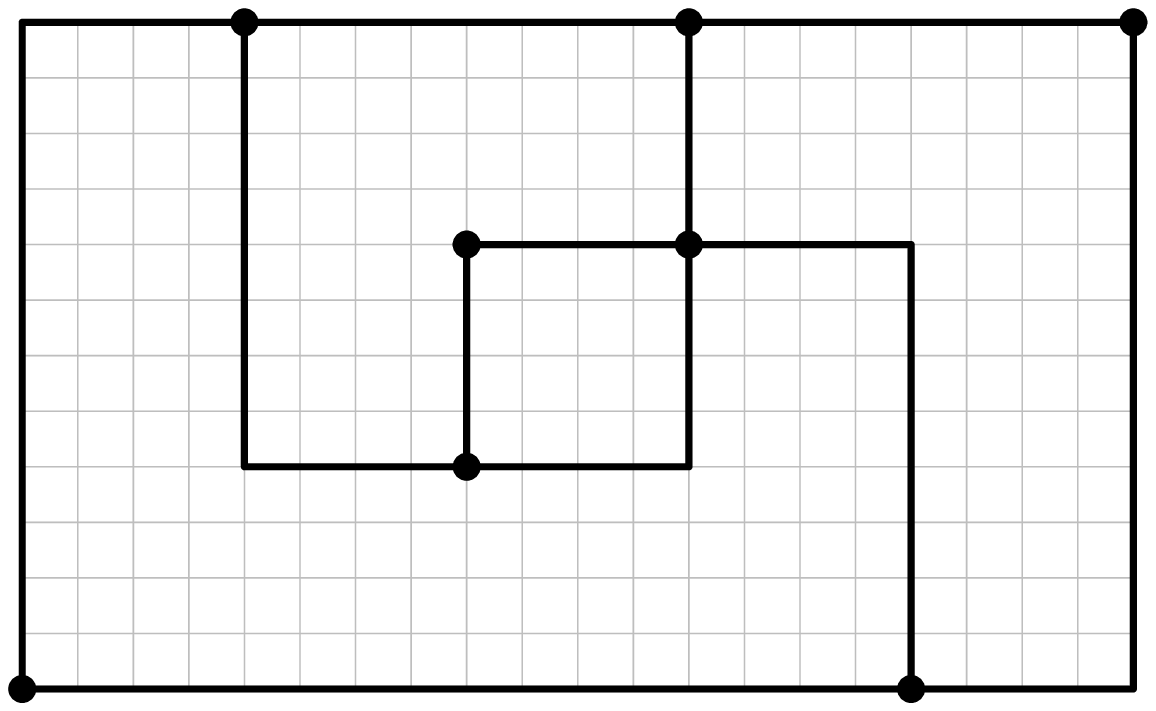}&&\includegraphics[scale=0.5]{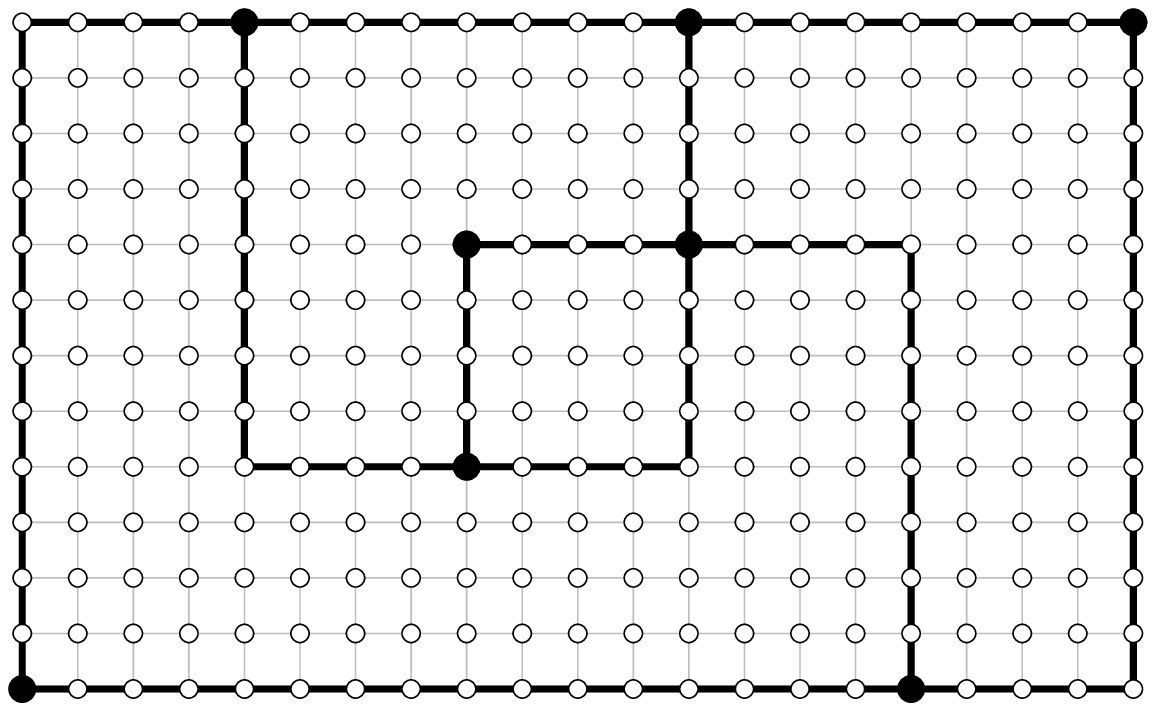}
\end{tabular}
\caption{Embedding a planar graph in a grid.}
\label{fig:embed}
\end{figure}
Note that for any integer $k$, we can embed $G$ in a fine enough grid so that
every horizontal and vertical segment in the original drawing of $G$
becomes a path that contains at least $k$ vertices in $G'$. In
Figure~\ref{fig:embed}, we have chosen $k=5$.

Let $R(w,h)$ be the rectangular grid of height $h$ and width $w$.
A path in $R(w,h)$ that contains only vertices from one row of the grid is called
a \emph{horizontal grid-path} and one that contains vertices from only one
column is called a \emph{vertical grid-path}. We denote
by $R'(w,h)$ the graph obtained by subdividing each edge of $R(w,h)$ once
and by adding paths of length 3 between the newly introduced vertices as shown
in Figure~\ref{fig:subdiv}.

\begin{figure}
\center
\includegraphics[scale=.6]{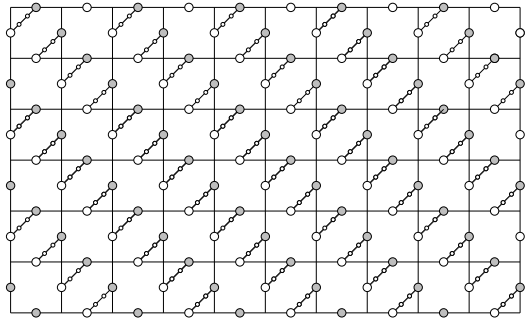}
\caption{The graph $R'(11,7)$. The vertices of the grid are not shown.}
\label{fig:subdiv}
\end{figure}
 
\begin{lemma}
\label{lem:R'}
Any planar graph $G$, on $n$ vertices and of maximum degree
4, has an even subdivision that is an induced subgraph of $R'(w,h)$ for
some values of $w$ and $h$ that are linear in $n$.
\end{lemma}
\begin{proof}
Let $H$ be the subdivision of $G$ that is an induced subgraph of the
grid $R(w,h)$. Let $P_e$ denote the path in $H$ corresponding to an
edge $e$ in $G$. We assume that $P_e$ is the union of horizontal and
vertical grid-paths of length at least 5.  We now transform the grid
$R(w,h)$ into $R'(w,h)$ by subdividing each edge once and by adding
paths of length 3 between the newly introduced vertices as explained
before. Clearly, a 1-subdivision of $H$, which we shall denote by
$H'$, is an induced subgraph of $R'(w,h)$. It is also clear that $H'$
is an odd subdivision of $G$. Let $P'_e$ denote the path in $H'$
corresponding to an edge $e$ of $G$. Note that $P'_e$ consists of
1-subdivisions of vertical and horizontal grid-paths.

For every edge $e$ of $G$, we do the following procedure on $P'_e$ in
$H'$ to obtain a new graph $H''$: we replace one of the subdivided
horizontal or vertical grid-paths that make up $P'_e$ to obtain
$P''_e$ which has an even number of vertices as shown in Figure
\ref{fig:modify}. The new graph $H''$ so obtained is an even
subdivision of $G$ and is also an induced subgraph of
$R'(w,h)$.\hfill\bbox
\end{proof}

\begin{figure}
\center
\begin{tabular}{cp{0.1in}c}
\includegraphics[scale=0.6]{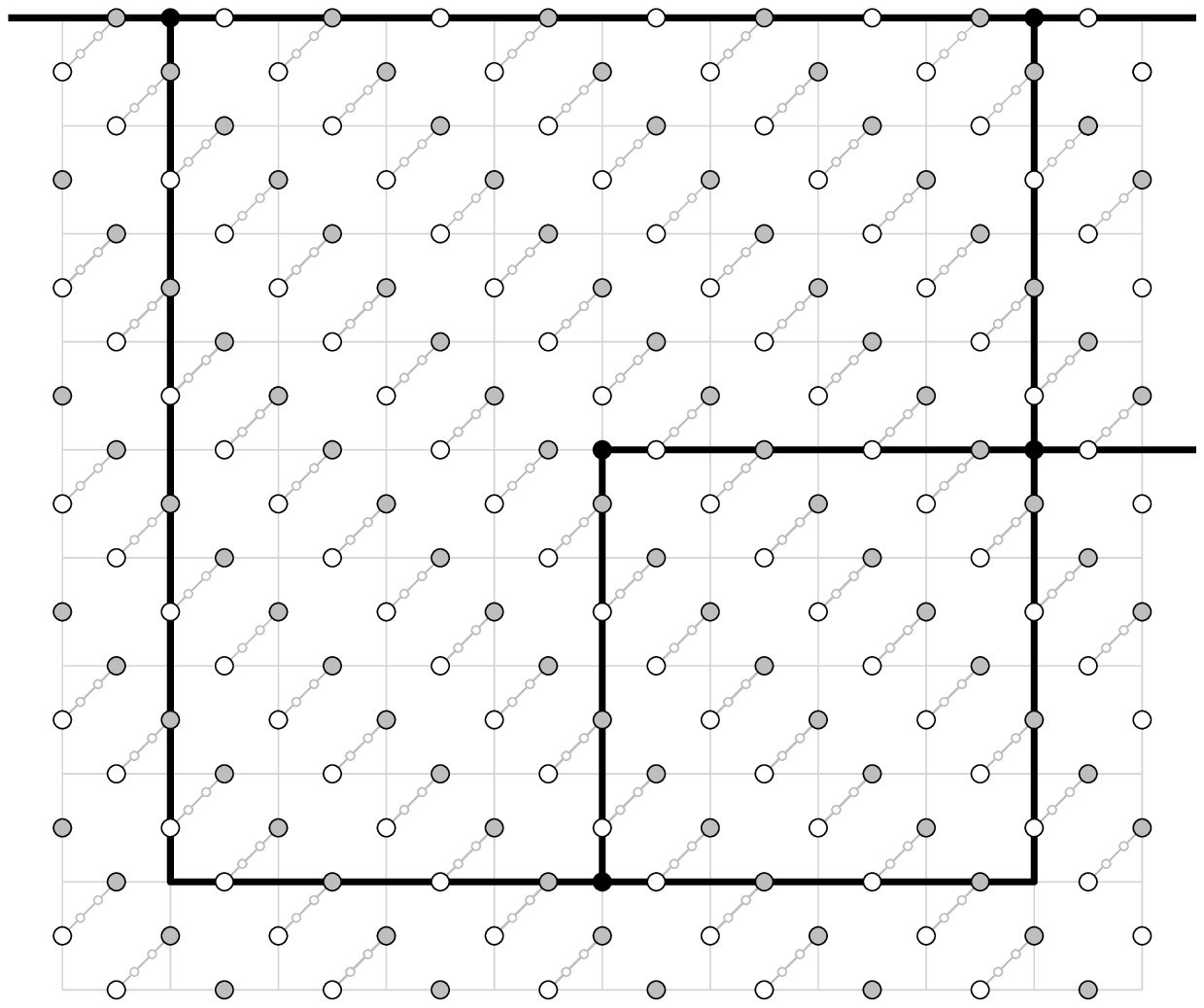}&&
\includegraphics[scale=0.6]{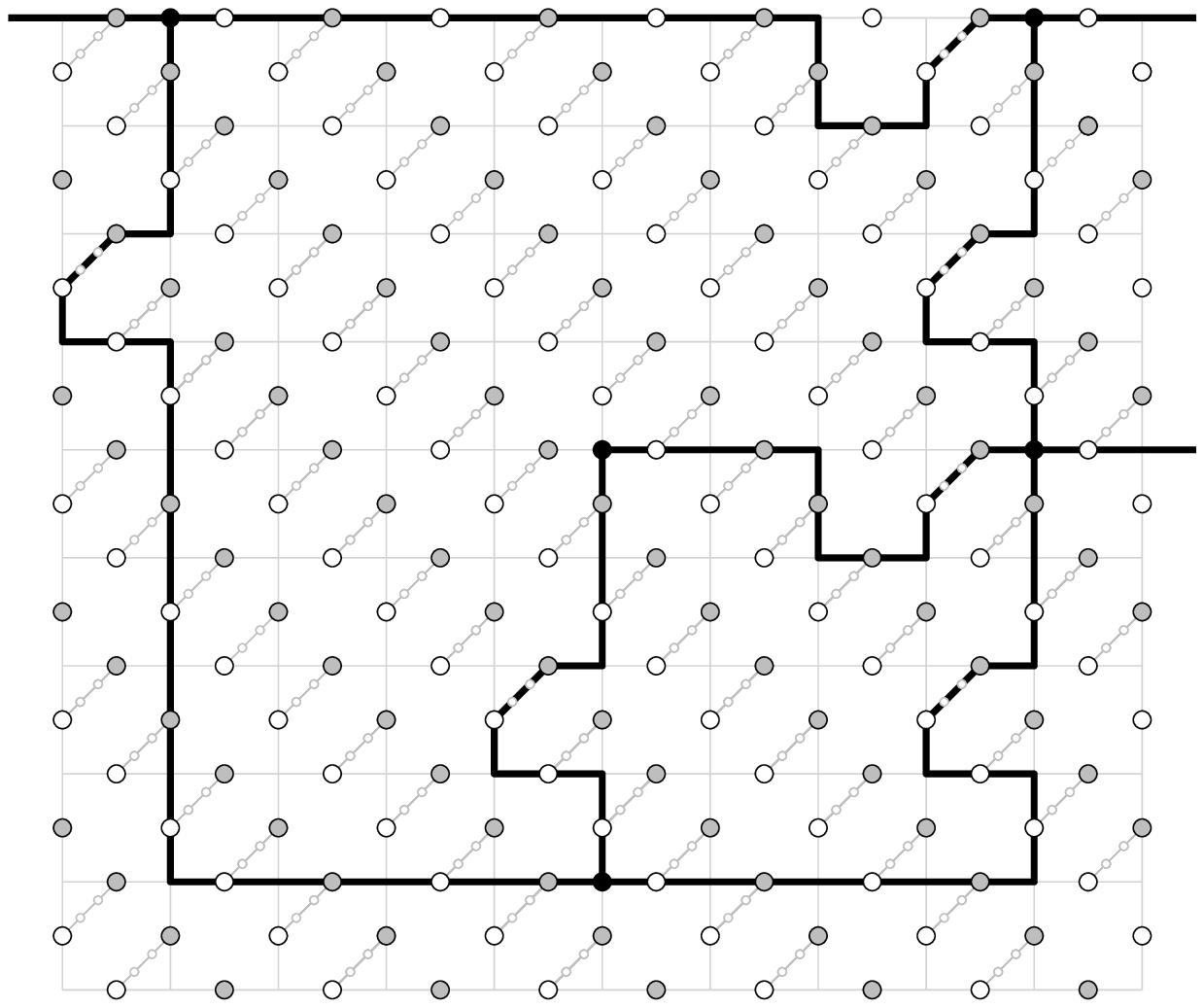}
\end{tabular}
\caption{Modifying the paths in $H'$ to obtain $H''$: A part of the graph in
Figure~\ref{fig:embed} is shown. The grid vertices are not drawn.}
\label{fig:modify}
\end{figure}

\begin{lemma}
\label{lem:u2i-u3t}
For any $w$ and $h$ the graph $\overline{R'(w,h)}$ is both a unit
2-interval graph as well as a unit 3-track graph.  Thus since those
classes are closed under taking induced subgraphs, they also contain
the induced subgraphs of $\overline{R'(w,h)}$.
\end{lemma}
\begin{proof}
The graph $Q(w,l)$ is defined as follows.  $V(Q(w,l))=X^o\cup X^e\cup
A\cup B\cup C\cup D$ where $X^o=\{x^o_1,\ldots,x^o_{w(l+1)}\}$,
$X^e=\{x^e_1,\ldots,x^e_{wl}\}$, $A=\{a_1,\ldots,a_{2wl}\}$,
$B=\{b_1,\ldots,b_{2wl}\}$, $C=\{c_1,\ldots,c_{2wl}\}$ and
$D=\{d_1,\ldots,d_{2wl}\}$.

\begin{eqnarray*}
E(Q(w,l))&=&\bigcup_{i=1}^{wl}\{x^o_ia_{2i-1}, x^o_ia_{2i}\}
\cup \bigcup_{i=w+1}^{w(l+1)}\{x^o_ib_{2(i-w)-2},x^o_ib_{2(i-w)-1}\}\\
&&\cup \bigcup_{i=1}^{wl}\{x^e_ia_{2i-2},x^e_ia_{2i-1},x^e_ib_{2i-1},x^e_ib_{2i}\}\\
&&\cup \bigcup_{i=1}^{2wl-1}\{a_ic_i,c_id_i,d_ib_{i+1}\}
\end{eqnarray*}

Figure \ref{fig:auggrid} shows a drawing of the graph $Q(w,l)$. The
vertices in $(\bigcup_{i=1}^l \{a_{2wi},b_{2wi}\}\cup
\bigcup_{i=1}^{2wl}\{c_i,d_i\})$ are not shown to avoid clutter.
\begin{figure}
\center
\includegraphics[scale=.75]{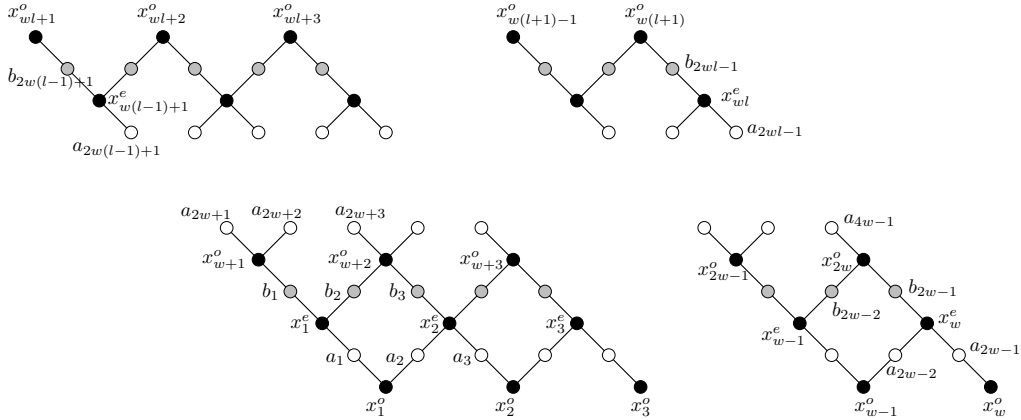}
\caption{Part of the graph $Q(w,l)$.}
\label{fig:auggrid}
\end{figure}
It can be seen that $R'(w,h)$ is an induced subgraph of
$Q(w,\lceil\frac{w+h}{2}\rceil-1)$~(see Figure~\ref{fig:grid}).
Thus, to show that for any $w$ and $h$, $\overline{R'(w,h)}$ is a unit
2-interval graph and a unit 3-track graph, we only need to show that
$\overline{Q(w,l)}$ for any $w$ and $l$ is both a unit 2-interval graph
as well as a unit 3-track graph.
\begin{figure}
\center
\includegraphics[scale=0.5]{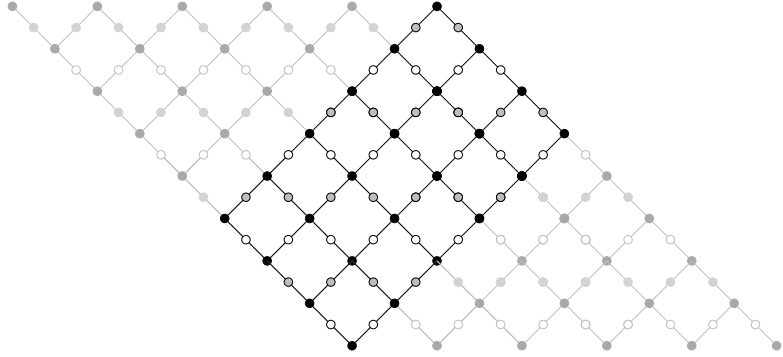}
\caption{$R'(w,h)$ is an induced subgraph of $Q(w,l)$ where $l=\lceil
\frac{w+h}{2}\rceil-1)$. The vertices in $(\bigcup_{i=1}^l \{a_{2wi},
b_{2wi}\}\cup \bigcup_{i=1}^{2wl}\{c_i,d_i\})$ are not shown.}
\label{fig:grid}
\end{figure}

We construct a unit 2-interval representation $f$ for $\overline{Q(w,l)}$ as follows (see also Figure~\ref{fig:npc-u2i}).

\medskip

$
\begin{array}{lcl}
I_1(a_i)&=&[2i,2i+6n]\\
I_1(c_i)&=&[2i+6n+2,2i+12n+2]\\
I_1(x^o_i)&=&[4i+6n+1,4i+12n+1]\\
I_2(a_i)&=&[2i+12n+4,2i+18n+4]\\
I_1(x^e_i)&=&[4i+6n-1,4i+12n-1]\\
I_1(d_i)&=&[2i,2i+6n]\\
I_1(b_i)&=&[18n+6-2i,24n+6-2i]\\
I_2(d_i)&=&[24n+6-2i,30n+6-2i]\\
I_2(b_i)&=&[30n+10-2i,36n+10-2i]\\
I_2(x^e_i)&=&[24n+9-4i,30n+9-4i]\\
I_2(x^o_i)&=&[24n+11-4i+4w,30n+11-4i+4w]\\
I_2(c_i)&=&[30n+8-2i,36n+8-2i]
\end{array}
$

\medskip
\begin{figure}
\center
\includegraphics[scale=.75]{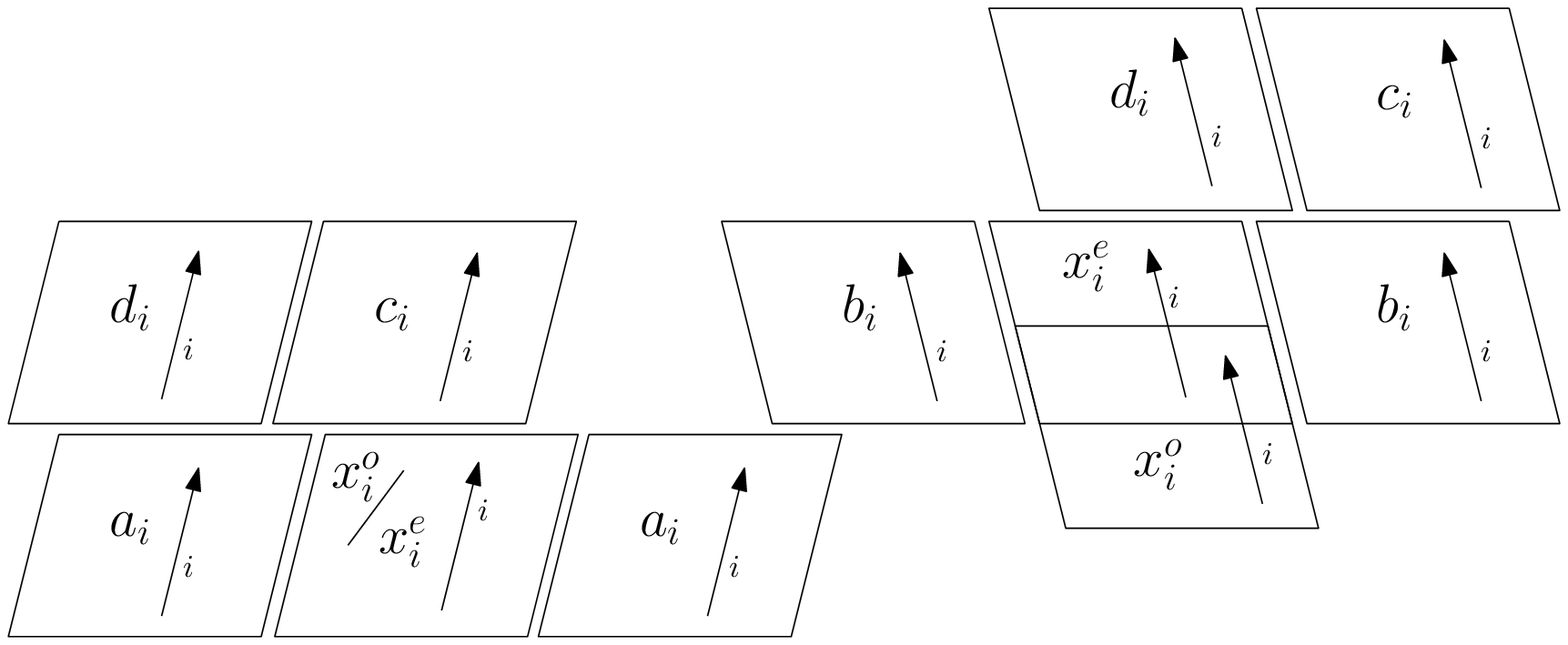}
\caption{Unit 2-interval representation of $\overline{Q(w,l)}$.}
\label{fig:npc-u2i}
\end{figure}

It is easy to verify that all the intervals have length $6n$.  Then
one can also check in the figure that this is a valid unit 2-interval
representation of $\overline{Q(w,l)}$. Note that this construction is
slightly more involved than the previous ones. Here the second blocks
$x^e_i$ and $x^o_i$ are slightly shifted. This is due to the fact that
$x^e_i$ is adjacent to every $b_j$, except $b_i$ and $b_{i+1}$, and
that $x^o_i$ has to avoid a distinct $b$'s.  Indeed $x^o_i$ is
adjacent to every $b_j$, except $b_{i-w-1}$ and $b_{i-w}$.  We
construct now a unit 3-track representation for $\overline{Q(w,l)}$ as
follows (see also Figure~\ref{fig:npc-u3t}).

\medskip

$
\begin{array}{lcl}
I_1(b_i)&=&[i,i+2n]\\
I_1(x^e_i)&=&[2i+2n+1,2i+4n+1]\\
I_1(x^o_i)&=&[2i-2w+2n,2i-2w+4n]\\
I_1(a_i)&=&[i+4n+4,i+6n+4]\\
I_1(c_i)&=&[i+2n+3,i+4n+3]\\
I_1(d_i)&=&[i+4n+4,i+6n+4]\\
I_2(a_i)&=&[i,i+2n]\\
I_2(x^e_i)&=&[2i+2n,2i+4n]\\
I_2(x^o_i)&=&[2i+2n+1,2i+4n+1]\\
I_2(b_i)&=&[i+2w+4n+4,i+2w+6n+4]\\
I_2(d_i)&=&[i+2w+2n+4,i+2w+4n+4]\\
I_2(c_i)&=&[i+2w+4n+5,i+2w+6n+5]\\
I_3(x^o_i)&=&[2i,2i+4n]\\
I_3(a_i)&=&[i+4n+2,i+8n+2]\\
I_3(c_i)&=&[i+8n+3,i+12n+3]\\
I_3(b_i)&=&[8n+3-i,12n+3-i]\\
I_3(d_i)&=&[4n+1-i,8n+1-i]\\
I_3(x^e_i)&=&[12n+5-2i,16n+5-2i]
\end{array}
$

\medskip
\begin{figure}
\center
\includegraphics[scale=.75]{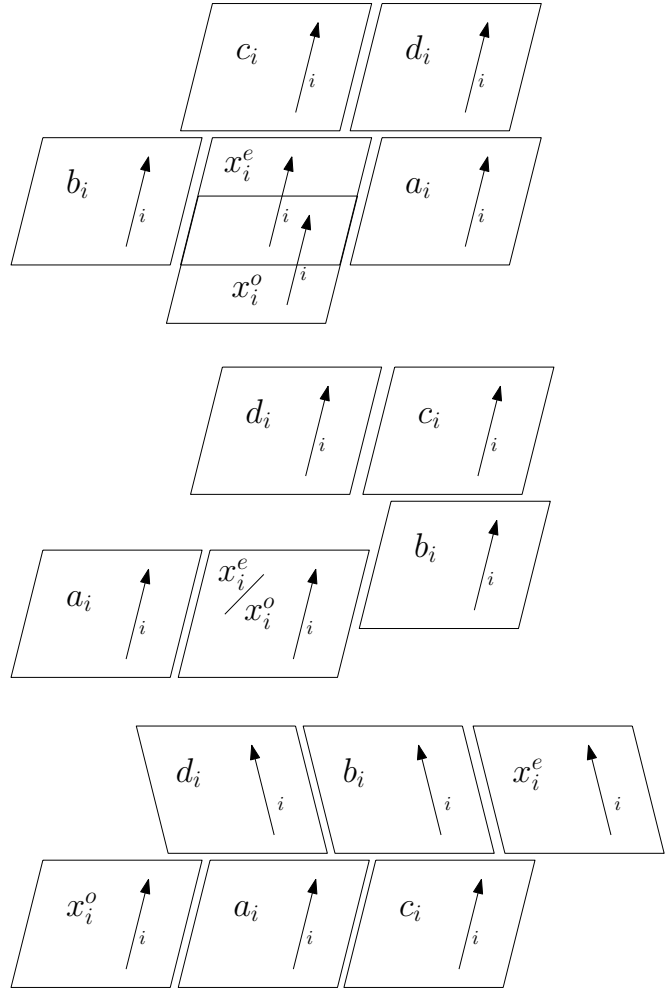}
\caption{Unit 3-track representation of $\overline{Q(w,l)}$.}
\label{fig:npc-u3t}
\end{figure}

It is easy to verify that all the intervals have length $4n$.  Then
one can also check in the figure that this is a valid unit 3-track
representation of $\overline{Q(w,l)}$. Note that here also one has to be carefull
with the $b_i$'s that $x^e_i$ and $x^o_i$ have to avoid.\hfill\bbox
\end{proof}

\begin{theorem}
\label{thm:np-c}
MAXIMUM CLIQUE is NP-complete for unit 2-interval and unit 3-track graphs.
\end{theorem}
\begin{proof}
It is known that the MAXIMUM INDEPENDENT SET problem is NP-complete
even when restricted to planar graphs of degree at most
3~\cite{Garey}.  It is folklore that the instance $(G,k)$ of MAXIMUM
INDEPENDENT SET is equivalent to an instance $(H,k+k')$, where $H$ is
an even subdivision of $G$ with $|V(G)|+2k'$ vertices.  Thus according
to Lemma~\ref{lem:R'}, MAXIMUM INDEPENDENT SET is NP-complete on the
class of induced subgraphs of $R'(w,h)$.  MAXIMUM CLIQUE is thus
NP-complete on the class of induced subgraphs of $\overline{R'(w,h)}$.
Finally by Lemma~\ref{lem:u2i-u3t} this class of graphs is contained
in unit 2-interval and unit 3-track graphs. MAXIMUM CLIQUE is thus
NP-complete on these classes.\hfill\bbox
\end{proof}

\section{Complements of $t$-interval graphs}
\label{sec:comp}
Very recently, Jiang and Zhang studied the class of complements of
$t$-interval graphs~\cite{JiangZhang}. In particular they proved that
MINIMUM (INDEPENDENT) DOMINATING SET parameterized by the solution
size is in W[1] for co-$2$-interval graphs, and they proved that
MINIMUM DOMINATING SET is W[1]-hard for co-$3$-track graphs.

Following the same line of proof as for Theorem~\ref{thm:apx} we
can prove the following APX-hardness results, for this kind of graph classes.
\begin{theorem}$\ $
\begin{itemize}
\item[(i)] MINIMUM VERTEX COVER is APX-complete in co-2-interval
  graphs, and the complement classes of all the classes of
  Theorem~\ref{thm:subd}.
\item[(ii)] For any graph $G$, $Subd_3(G)$ is a co-2-interval, a
  co-unit-3-interval, a co-3-track, a co-unit-4-track, and a
  co-2-circular track graph, and MINIMUM (INDEPENDENT) DOMINATING SET
  is APX-hard for these classes of graphs.
\end{itemize}
\end{theorem}
\begin{proof}
The first item follows from the fact that MINIMUM VERTEX COVER is
2-approximable~\cite{Monien} and the first item of the following
theorem.
\begin{theorem}[\cite{ChlebChleb}]\label{thm:apxhardbis}
\begin{itemize}
\item[(i)] MINIMUM VERTEX COVER is APX-complete when restricted to
  $2k$-subdivisions of 3-regular graphs for any fixed integer $k\ge
  0$.
\item[(ii)] The problems MINIMUM DOMINATING SET, and MINIMUM
  INDEPENDENT DOMINATING SET are APX-complete when restricted to
  $3k$-subdivisions of degree at most 3 graphs for any fixed integer
  $k\ge 0$.
\end{itemize}
\end{theorem}

For the second item the constructions are described in the following
figures.  Here an edge $e_k=x_ix_j$ of $G$, with $i<j$, corresponds to
a path $(x_{l(k)},a_k,b_k,c_k,x_{r(k)})$ of $Subd_3(G)$. Then
Theorem~\ref{thm:apxhardbis}.(ii) clearly implies the APX-hardness.
\begin{figure}
\center
\includegraphics[scale=0.6]{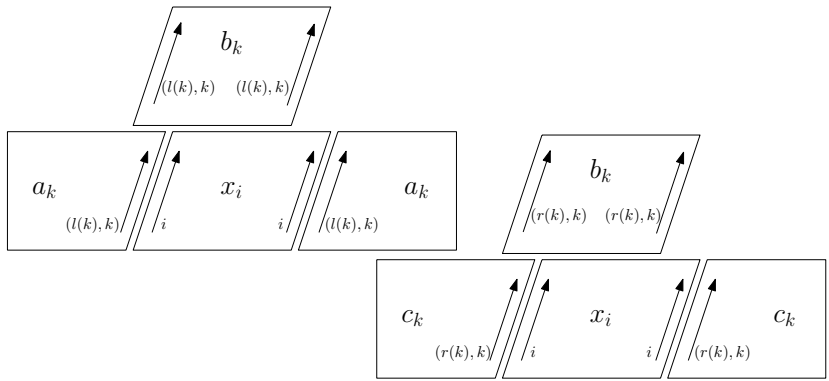}
\caption{The 2-interval representation of $\overline{Subd_3(G)}$.}
\label{fig:dc2i}
\end{figure}

\begin{figure}
\center
\includegraphics[scale=0.8]{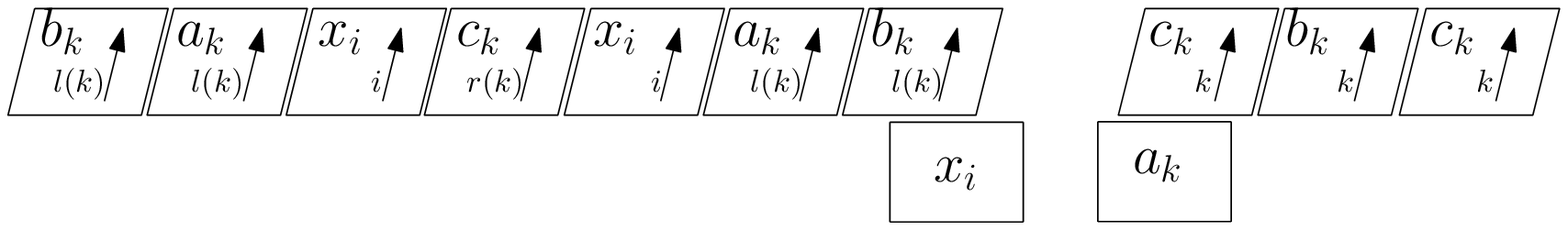}
\caption{The unit 3-interval representation of $\overline{Subd_3(G)}$.}
\label{fig:dcu3i}
\end{figure}

\begin{figure}
\center
\includegraphics[scale=0.6]{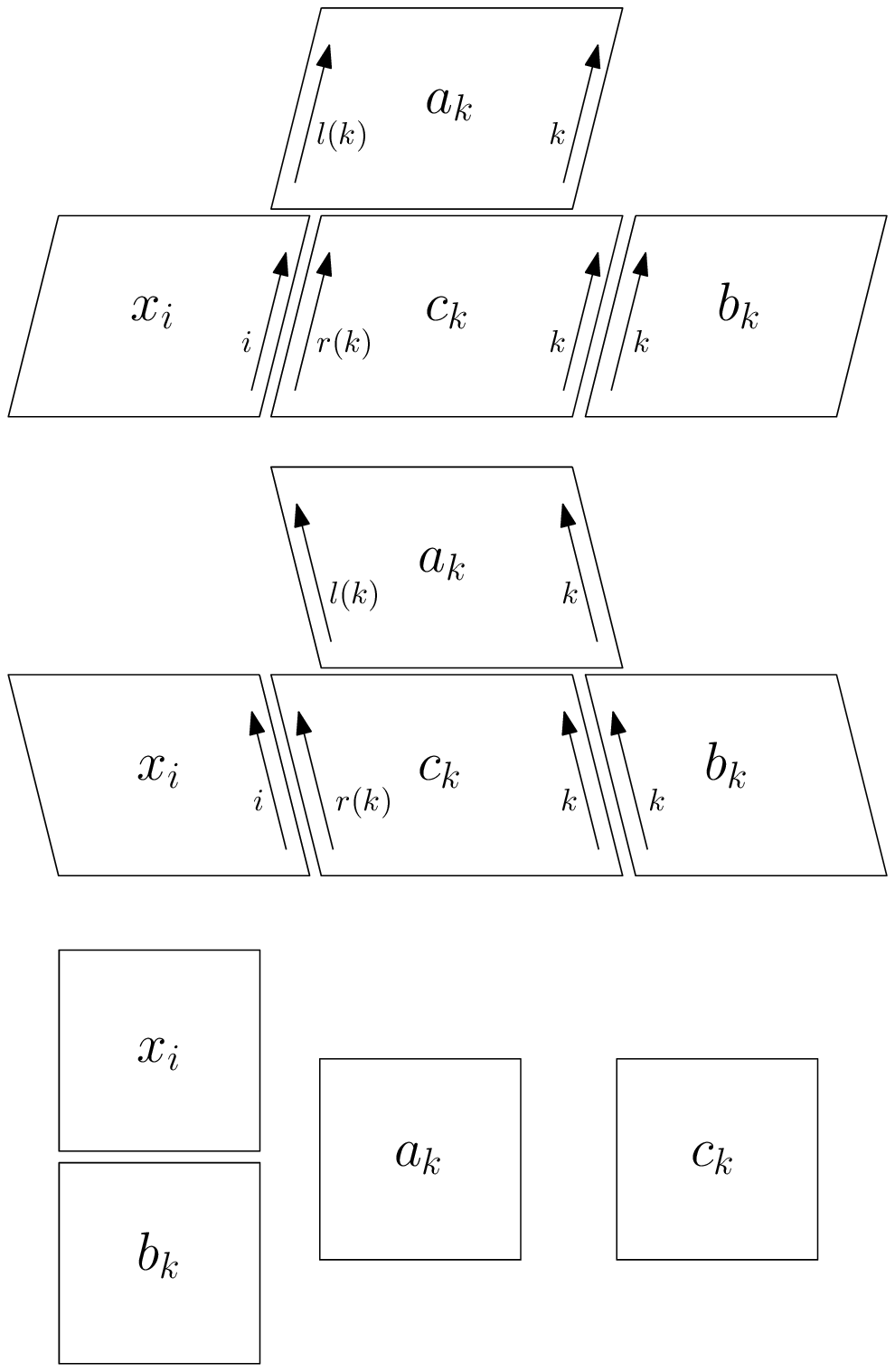}
\caption{The 3-track representation of $\overline{Subd_3(G)}$.}
\label{fig:dc3t}
\end{figure}

\begin{figure}
\center
\includegraphics[scale=0.6]{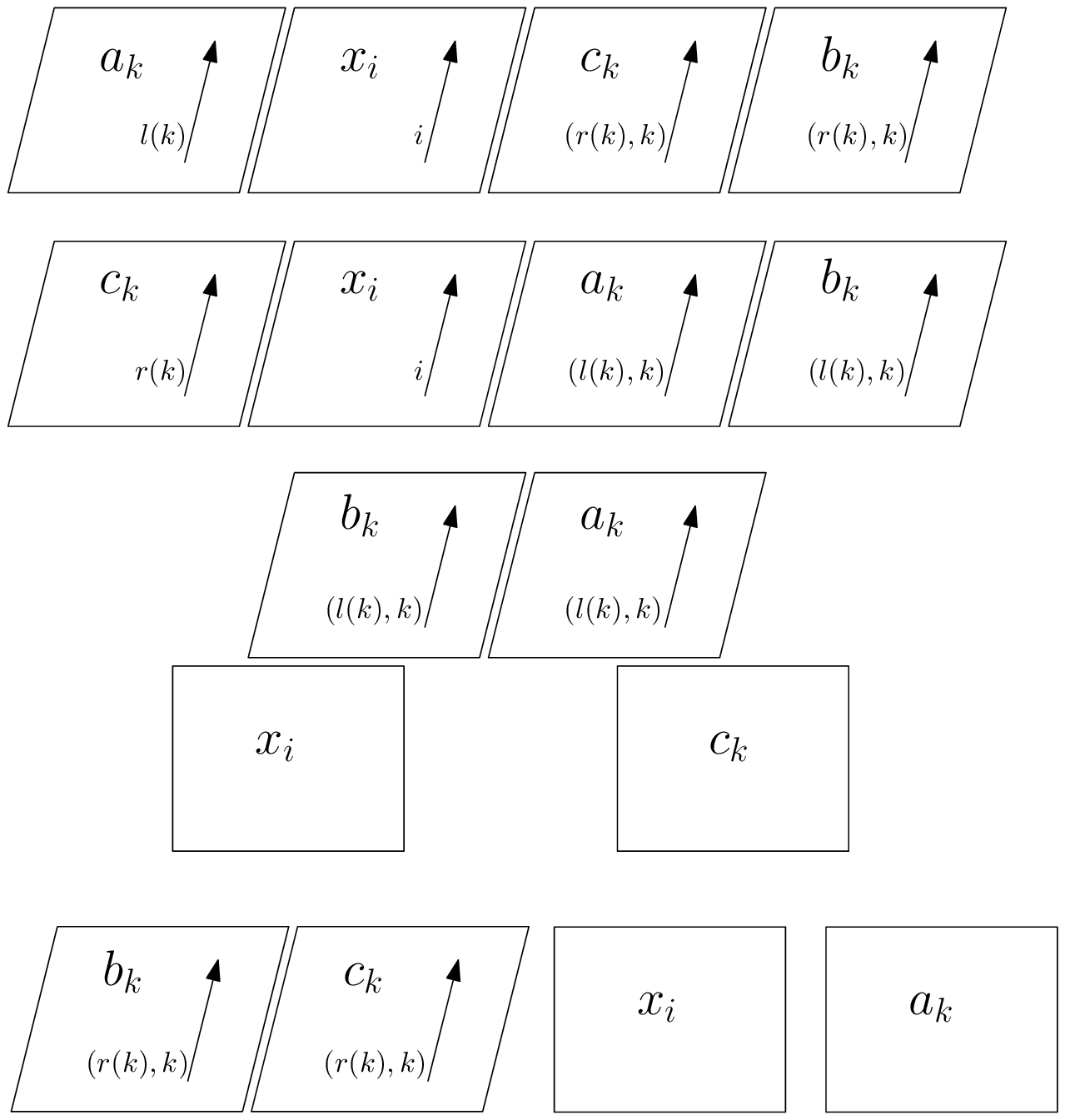}
\caption{The unit 4-track representation of $\overline{Subd_3(G)}$.}
\label{fig:dcu4t}
\end{figure}

\begin{figure}
\center
\includegraphics[scale=0.6]{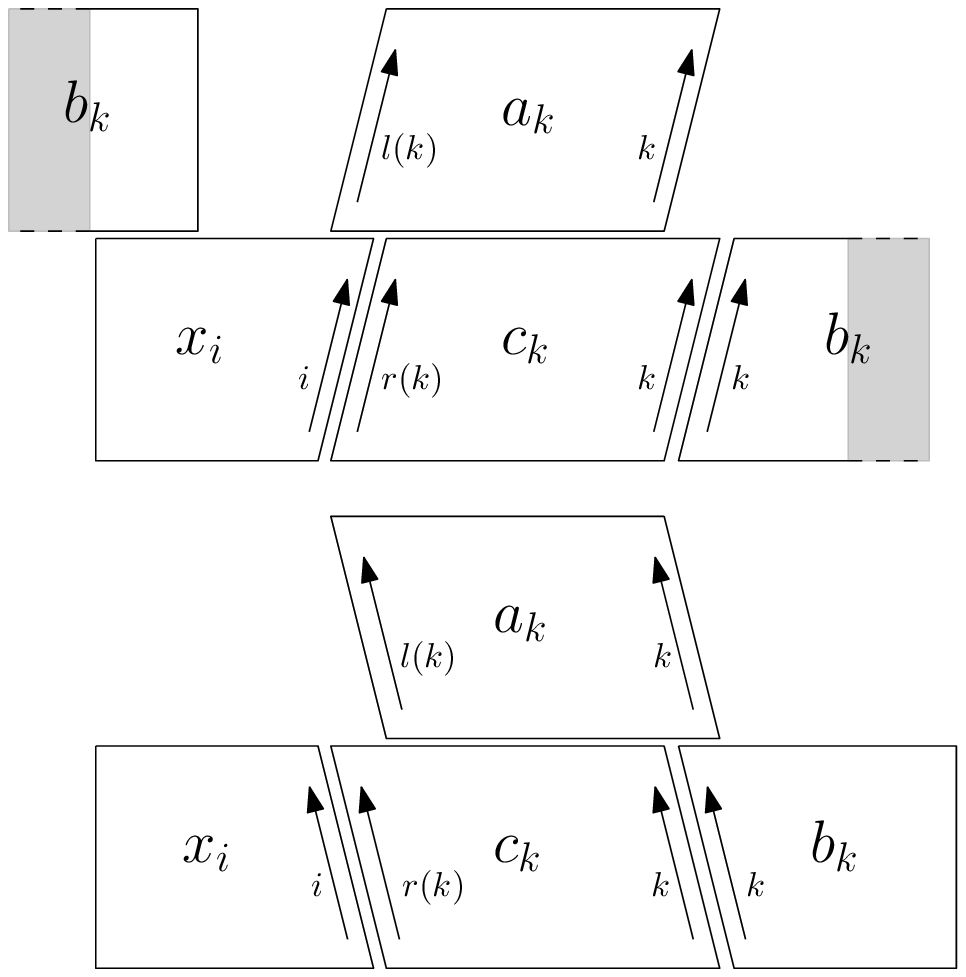}
\caption{The 2-circular track representation of $\overline{Subd_3(G)}$.}
\label{fig:dc2ct}
\end{figure}

\hfill\bbox
\end{proof}

\section{Concluding remarks}

The difference between the $4t$-approximation of Kammer et
al.~\cite{Kammer} and our $t$-approximation lies in two places. In
their paper they proved that $t$-interval graphs are $2t$-perfectly
orientable, but following the lines of Theorem~\ref{thm:approx} one
can see that those graphs are $t$-perfectly orientable. This improves
their approximation for MAXIMUM WEIGHTED INDEPENDENT SET, MINIMUM
VERTEX COLORING, and MINIMUM CLIQUE PARTITION in $t$-interval
graphs. For MAXIMUM WEIGHTED INDEPENDENT SET and MINIMUM VERTEX
COLORING this reaches the best known ratio of $t$~\cite{Bar} in a
simpler way, and for the other problems it improves the best known
approximation ratios~\cite{Kammer}.  Then Kammer et al. proved that
MAXIMUM WEIGHTED CLIQUE can be $2k$-approximated in $k$-perfectly orientable
graphs. Again, following the lines of Theorem~\ref{thm:approx} one can
see that MAXIMUM WEIGHTED CLIQUE can be $k$-approximated for those graphs. This
improves (by 2) their approximation for MAXIMUM WEIGHTED CLIQUE in $t$-fat
objects intersection graphs.

In our approximation algorithm (as in the previous algorithms) we
assume that we are given an interval representation. We wonder what we
can do if we are not given such representation.
\begin{open}
Can MAXIMUM (WEIGHTED) CLIQUE be polynomially $c(t)$-approximated in
$t$-interval graphs, for some function $c$, if we are not given an
interval representation?
\end{open}
 This would be the case if there is an algorithm that computes, given
 a $t$-interval graph $G$, a $c(t)$-interval representation of
 $G$. Actually even when we are given a representation, the
 approximation ratio might be far from the optimal.
\begin{open}
Does there exists an approximation algorithm for MAXIMUM (WEIGHTED)
CLIQUE in $t$-interval graphs with a better approximation ratio?
\end{open}
Let us call $f(t)$ the better ratio a polynomial algorithm can achieve
on $t$-interval graphs (actually $f(t)$ should be an infimum).  For
any graph $G$ on $n$ vertices, it is easy to construct a $n$-interval
representation of $G$. Thus since for any $\epsilon > 0$, one cannot
$O(n^{1-\epsilon})$-approximate the MAXIMUM CLIQUE unless
P~=~NP~\cite{Zuckerman}, we certainly have
$f(t)=\Omega(t^{1-\epsilon})$.

The current status of the complexity of the MAXIMUM CLIQUE problem for
the various classes of multiple interval graphs that were studied are
shown in the table below (where ``Unres.''  stands for
``Unrestricted'').

\medskip

\noindent\begin{tabular}{|p{.06\textwidth}|p{.11\textwidth}|p{.11\textwidth}|p{.11\textwidth}|p{.11\textwidth}|p{.11\textwidth}|p{.11\textwidth}|p{.11\textwidth}|p{.11\textwidth}|}
\hline
\multirow{2}{*}{\parbox{.06\textwidth}{\centering $t$}}&\multicolumn{2}{|c|}{$t$-track}&\multicolumn{2}{|c|}{$t$-interval}&\multicolumn{2}{|c|}{Circular $t$-track}&
\multicolumn{2}{|c|}{Circular $t$-interval}\\
\cline{2-9}
&\centering Unit&\centering Unres.&\centering Unit&\centering Unres.&\centering Unit&\centering Unres.&\centering Unit&{\centering Unres.}\\
\hline
\parbox{.06\textwidth}{\centering 1}&
\parbox{.11\textwidth}{\centering P}&\parbox{.11\textwidth}{\centering P}&
\parbox{.11\textwidth}{\centering P}&\parbox{.11\textwidth}{\centering P}&
\parbox{.11\textwidth}{\centering P}&\parbox{.11\textwidth}{\centering P}&
\parbox{.11\textwidth}{\centering P}&\parbox{.11\textwidth}{\centering P}\\
\hline
\parbox{.06\textwidth}{\centering 2}&
\parbox{.11\textwidth}{\centering P}&\parbox{.11\textwidth}{\centering P}&
\parbox{.11\textwidth}{\centering NP-c}&\parbox{.11\textwidth}{\centering APX-c}&
\parbox{.11\textwidth}{\centering ?}&\parbox{.11\textwidth}{\centering APX-c}&
\parbox{.11\textwidth}{\centering APX-c}&\parbox{.11\textwidth}{\centering APX-c}\\
\hline
\parbox{.06\textwidth}{\centering 3}&
\parbox{.11\textwidth}{\centering NP-c}&\parbox{.11\textwidth}{\centering APX-c}&
\parbox{.11\textwidth}{\centering APX-c}&\parbox{.11\textwidth}{\centering APX-c}&
\parbox{.11\textwidth}{\centering NP-c}&\parbox{.11\textwidth}{\centering APX-c}&
\parbox{.11\textwidth}{\centering APX-c}&\parbox{.11\textwidth}{\centering APX-c}\\
\hline
\parbox{.06\textwidth}{\centering $\geq 4$}&
\parbox{.11\textwidth}{\centering APX-c}&\parbox{.11\textwidth}{\centering APX-c}&
\parbox{.11\textwidth}{\centering APX-c}&\parbox{.11\textwidth}{\centering APX-c}&
\parbox{.11\textwidth}{\centering APX-c}&\parbox{.11\textwidth}{\centering APX-c}&
\parbox{.11\textwidth}{\centering APX-c}&\parbox{.11\textwidth}{\centering APX-c}\\
\hline
\end{tabular}

\medskip

The blanks in this table clearly imply the following questions.
\begin{open}
Is MAXIMUM CLIQUE for unit 2-interval graphs, unit 3-track graphs
or unit 3-circular track graphs APX-hard, or does it admit a PTAS?
\end{open}

\begin{open}
Is MAXIMUM CLIQUE for unit 2-circular track graphs Polynomial or NP-complete?
\end{open}
Koenig~\cite{Koenig} explains that 2-track graphs have a
polynomial-time algorithm for MAXIMUM CLIQUE because for any 2-track
representation of a clique, there is a transversal of size 2 (i.e. two
points such that for every vertex, at least one of its intervals
contains one of these points). We note that this is not true for unit
2-circular track graphs as the complete graph on 5 vertices has a unit
2-circular track representation in which each circular track induces a
cycle on 5 vertices.  This representation clearly does not have a
transversal of size 2.



\begin{thebibliography}{10}


\bibitem{Asi}
Andrei Asinowski, Elad Cohen, Martin C. Golumbic, Vincent Limouzy, Marina Lipshteyn, and Michal Stern.
\newblock Vertex Intersection Graphs of Paths on a Grid
\newblock {\em Journal of Graph Algorithms and Applications}, 16 (2): 129--150, 2012.

\bibitem{Aumann}
Yonatan Aumann, Moshe Lewenstein, Oren Melamud, Ron Y. Pinter, and Zohar Yakhini.
\newblock Dotted interval graphs and high throughput genotyping. 
\newblock In {\em Proc. of the 16th Annual Symposium on Discrete
  Algorithms}, SODA '05, 339--348, 2005.

\bibitem{Bar}
Reuven Bar-Yehuda, Magn\'{u}s M. Halld\'{o}rsson, Joseph S. Naor, Hadas Shachnai, and Irina Shapira. 
\newblock Scheduling split intervals. 
\newblock {\em SIAM J. Comput.} 36:1--15, 2006.

\bibitem{BermanFujito}
Piotr Berman, and Toshihiro Fujito.
\newblock On approximation properties of the Independent set problem for degree 3 graphs.
\newblock In {\em  Proceedings of the 4th Workshop on Algorithms and Data Structures}, WADS '95, LNCS Vol. 955, Springer-Verlag, pages 449--460, 1995. 


\bibitem{Butman}
Ayelet Butman, Danny Hermelin, Moshe Lewenstein, and Dror Rawitz.
\newblock Optimization problems in multiple-interval graphs.
\newblock In {\em Proceedings of the eighteenth annual ACM-SIAM symposium on
  Discrete algorithms}, SODA '07, pages 268--277, 2007.

\bibitem{Cabello}
Sergio Cabello, Jean Cardinal, and Stefan Langerman.
\newblock The Clique Problem in Ray Intersection Graph.
\newblock {\em arXiv} \url{http://arxiv.org/pdf/1111.5986.pdf}, Nov. 2011.

\bibitem{ChlebChleb}
Miroslav Chleb\'ik and Janka Chleb\'ikova.
\newblock The complexity of combinatorial optimization problems on
  $d$-dimensional boxes.
\newblock {\em {SIAM} Journal on Discrete Mathematics}, 21(1):158--169, 2007.

\bibitem{Crochemore}
Maxime Crochemore, Danny Hermelin, Gad Landau, Dror Rawitz, and St\'ephane Vialette.
\newblock Approximating the 2-interval pattern problem. 
\newblock In {\em Proc. of the 13th Annual European Symposium on Algorithms}, ESA '05, 426--437, 2005.

\bibitem{Garey} 
Michael R. Garey and David S. Johnson.  
\newblock Rectilinear steiner tree problem is NP-complete.  
\newblock {\em SIAM J. Appl. Math.} 6: 826--834, 1977.

\bibitem{Gavril73}
Fanica Gavril.
\newblock Algorithms for a maximum clique and a maximum independent set of a circle graph.
\newblock {\em Networks}, 3: 261--273, 1973.

\bibitem{Gavril00}
Fanica Gavril.
\newblock Maximum weight independent sets and cliques in intersection graphs of filaments.
\newblock {\em Information Processing Letters}, 73(5–6):181--188, 2000.




\bibitem{Hochbaum}
Dorit S. Hochbaum, and Asaf Levin.
\newblock Cyclical scheduling and multi-shift scheduling: Complexity and approximation algorithms.
\newblock {\em Disc. Optimiz.} 3(4):327–340, 2006.


\bibitem{Hsu}
Wen-Lian Hsu.
\newblock Maximum weight clique algorithms for circular-arc graphs and circle graphs.
\newblock {\em SIAM J. Comput.} 14(1):224--231, 1985.

\bibitem{JiangZhang}
Minghui Jiang and Yong Zhang
\newblock Parameterized Complexity in Multiple-Interval Graphs: Domination, Partition, Separation, Irredundancy.
\newblock {\em arXiv}, \url{http://arxiv.org/pdf/1110.0187v1.pdf}, Oct 2011.

\bibitem{Kaiser}
Tom{\'a}{\v{s}} Kaiser.
\newblock Transversals of d-Intervals.
\newblock {\em Discrete Comput Geom} 18:195--203, 1997.

\bibitem{Kammer}
Frank Kammer, Torsten Tholey, and Heiko Voepel.
\newblock Approximation algorithms for intersection graphs.
\newblock In {\em Proceedings of the 13th international workshop on
  Approximation Algorithms for Combinatorial Optimization Problems and 14th
  International workshop on Randomization and Computation}, APPROX/RANDOM'10,
  pages 260--273, Berlin, Heidelberg, 2010. Springer-Verlag.

\bibitem{Koenig}
Felix~G. K\"onig.
\newblock {\em Sorting with objectives}.
\newblock PhD thesis, Technische Universit\"at Berlin, 2009.

\bibitem{Krat} 
Jan Kratochv\'{i}l and Jaroslav Ne\v{s}et\v{r}il.  
\newblock Independent set and clique problems in intersection-defined classes
  of graphs.  
\newblock {\em Commentationes Mathematicae Universitatis Carolinae}, 31:85–93, 1990.


\bibitem{MiddPfeiff}
M.~Middendorf and F.~Pfeiffer.
\newblock The max clique problem in classes of string-graphs.
\newblock {\em Discrete Mathematics}, 108:365--372, 1992.

\bibitem{Monien}
Burkhard Monien and Ewald Speckenmeyer.
\newblock Ramsey numbers and an approximation algorithm for the vertex cover problem.
\newblock {\em Acta Inf.} 22, 115--123, 1985. 

\bibitem{PapaYanna}
Christos H. Papadimitriou, and Mihalis Yannakakis. 
\newblock Optimization, approximation, and complexity classes.
\newblock {\em J. Comput. System Sci.} 43:425--440, 1991. 

\bibitem{Scheiner1}
Edward~R. Scheinerman.
\newblock The maximum interval number of graphs with given genus.
\newblock {\em Journal of Graph Theory}, 11(3):441--446, 1987.

\bibitem{ScheinerWest}
Edward~R. Scheinerman and Douglas~B. West.
\newblock The interval number of a planar graph: Three intervals suffice.
\newblock {\em Journal of Combinatorial Theory, Series B}, 35(3):224--239,
  1983.

\bibitem{TrotHar}
William~T. Trotter and Frank Harary.
\newblock On double and multiple interval graphs.
\newblock {\em Journal of Graph Theory}, 3(3):205--211, 1979.

\bibitem{Valiant}
Leslie~G. Valiant.
\newblock Universality considerations in {VLSI} circuits.
\newblock {\em {IEEE} Transactions on Computers}, 30(2):135--140, 1981.

\bibitem{WestShmoys}
Douglas~B. West and David~B. Shmoys.
\newblock Recognizing graphs with fixed interval number is {NP}-complete.
\newblock {\em Discrete Applied Mathematics}, 8(3):295--305, 1984.

\bibitem{Zuckerman}
David Zuckerman.
\newblock Linear degree extractors and the inapproximability of max clique and chromatic number.
\newblock In {\em Proc. 38th ACM Symp. Theory of Computing}, STOC '06, 681–690, 2006.

\end{thebibliography}
\end{document}